\documentclass[11pt]{article}
\usepackage{amsmath,array,amssymb}
\usepackage[english]{babel}
\usepackage[T1]{fontenc}
\usepackage{inputenc}
\usepackage[a4paper,top=2.5cm,bottom=2.5cm,left=3cm,right=2cm]{geometry}
\usepackage{csquotes}

\usepackage{amssymb}
\usepackage{amsmath}
\usepackage{amsthm}
\usepackage{mathrsfs}
\usepackage{amssymb}
\usepackage{bm}
\usepackage{mathtools}
\usepackage{dsfont}

\usepackage{graphicx}
\usepackage{booktabs}
\usepackage{sectsty}
\usepackage{array}
\usepackage{caption}
\usepackage{subcaption}
\usepackage{calc}

\usepackage[affil-it]{authblk}

\newtheorem*{lem*}{Lemma}

\DeclareMathOperator{\spanned}{span}

\begin{document}
\title{Decoherence-free subspaces for a quantum register interacting with a spin environment}
\author{Pawe\l\ Nale\.zyty\footnote{e-mail: 242382@stud.umk.pl }$\ $ and Dariusz Chru\'sci\'nski}
\affil{Institute of Physics, Faculty of Physics, Astronomy and Informatics\\Nicolaus Copernicus University\\Grudziadzka 5, 87-100 Torun, Poland}

\maketitle

\begin{abstract}
We study a model of a quantum spin register interacting with an environment of spin particles in quantum-measurement limit. In the limit of collective decoherence we obtain the form of state vectors that constitute high-dimensional decoherence-free subspaces (DFS). In a more general setting we present sufficient and necessary conditions for existence of low-dimensional DFSs that can be used to construct subspaces of higher dimension.
\end{abstract}

\section{Introduction}

For over three decades the concept of quantum computers has tempted physicists with a promise of a tremendous decrease of computational time needed to solve certain problems in comparison with classical computers. The key element needed for proper work of a quantum computer \cite{jozsa:98} yet the hardest obstacle to overcome is the ability to maintain a coherent superposition of states. As environment-induced decoherence would cause errors in computation one must find effective ways of protection against it. Several methods have already been proposed, including quantum-error correcting codes \cite{knill:96}, dynamical decoupling \cite{viola:99} and encoding logical qubits in decoherence-free subspaces (DFS) \cite{lidar:98,zanardi:99a,zanardi:99b} which can also be used to protect quantum memory.

The theory of decoherence-free subspaces has been an area of intensive research for over 15 years now. Sufficient and necessary conditions for dynamics to support existence of DFSs has been found by Zanardi and Rasetti in \cite{zanardi:99b} for Hamiltonian approach to reduced dynamics and by Lidar et al in \cite{lidar:99,lidar:98} for operator-sum representation formalism and Markovian semigroup approach. Through group-theoretic considerations it has also been realised that symmetry of interaction plays an important role in arising of decoherence-free dynamics and causes the so-called multiple qubit errors \cite{lidar:01,zanardi:00}.

The model to be described in this paper represents a class of canonical models of decoherence known as spin-environment models \cite{schlosshauer:07}. It is a generalization of a central spin decoherence model considered in \cite{cucchietti:05} to a $K$-qubit spin register $\mathcal{R}$ interacting with $N$-spin environment $\mathcal{E}$ (with Hilbert spaces $\mathcal{H_R}$ and $\mathcal{H_E}$ respectively) via interaction Hamiltonian
\begin{equation}	\label{eq:hamiltonian_form}
	H_{int}=\frac{1}{2}\sum\limits_{i=1}^{K}\sigma_{z}^{(i)}\otimes\sum\limits_{j=1}^{N}g_{ij}\sigma_{z}^{(j)},
\end{equation}
where $g_{ij}$ quantify coupling strength of $i$-th register to $j$-th environmental spin and $\sigma_{z}^{(i)}$ acts as standard Pauli $\sigma_z$ operator on $i$-th register/environmental spin and as identity operator on all remaining ones.
It is clear that all properties of the model are encoded into the properties of the {\em interaction matrix} $G=[g_{ij}]$.

Here we will consider only quantum-measurement limit, in which the total Hamiltonian $H=H_{int}$. Spin-environment models are useful for modeling decoherence of physical systems in temperatures close to absolute zero, where the environment acts effectively as a collection of two-level systems rather than a bath of harmonic oscillators (see \cite{schlosshauer:07} and references thereof).

The structure of the paper is the following. In Section \ref{sec:collective_decoherence} we derive conditions for decoherence-free dynamics by explicit calculation of decoherence rates and consider the case of the most symmetric system-environment interaction leading to collective decoherence. In Section \ref{sec:general_case} the main result of general conditions for existence of a special kind of DF subspaces are obtained. The structure of register Hilbert space with emphasis on existing DFSs is then further explored in Section \ref{sec:discussion}.

\section{Decoherence-free dynamics condition and collective decoherence} \label{sec:collective_decoherence}

Due to a simple form of Hamiltonian \eqref{eq:hamiltonian_form} its eigenvectors can be found to be $$|\psi_{kn}\rangle=|k_1\ldots k_K\rangle|n_1\ldots n_N\rangle , $$ where $k_i,\,n_j=0,\,1$. We will denote register (environmental) states as $|k\rangle$ ($|n\rangle$) where $k$ ($n$) is decimal form of a number with binary representation given by a string $k_1\ldots k_K$ ($n_1\ldots n_N$). The corresponding eigenvalue $E_{kn}$ reads
\begin{equation}
	E_{kn}=\frac{1}{2}\sum\limits_{i=1}^{K}(-1)^{k_i}\cdot\sum\limits_{j=1}^{N}(-1)^{n_j}g_{ij}.
\end{equation}
Assume that initially the state of the composite system was separable:
\begin{equation}
	|\psi(0)\rangle=\sum\limits_{k=0}^{2^K-1}\sum\limits_{n=0}^{2^N-1}a_k b_n |k\rangle|n\rangle,
\end{equation}
At time $t$ it evolves into
\begin{equation}	\label{eq:evolved_state}
	|\psi(t)\rangle=\sum\limits_{k=0}^{2^K-1} a_k |k\rangle|\varepsilon_k(t)\rangle,
\end{equation}
where
\begin{equation}
	|\varepsilon_k(t)\rangle=\sum\limits_{n=0}^{2^N-1}e^{-iE_{kn}t} b_n|n\rangle.
\end{equation}

State \eqref{eq:evolved_state} is in general entangled, which at the level of the register leads to the evolution of the elements of the reduced density matrix $\rho(t)$ according to $\rho_{kk'}(t)=\rho_{kk'}(0)r_{kk'}(t)$, where the so-called decoherence rate $r_{kk'}(t)=\langle\varepsilon_{k'}(t)|\varepsilon_k(t)\rangle$ is   given by
\begin{equation}	\label{eq:decoherence_rate}
	r_{kk'}(t)=\sum\limits_{n=0}^{2^N-1}e^{-i(E_{kn}-E_{k'n})t}|b_n|^2.
\end{equation}
From the above equation one can deduce that no decoherence between states $|k\rangle$ and $|k'\rangle$ occurs iff $e^{-i(E_{kn}-E_{k'n})t}$ is $n$-independent, which translates into the following condition
$$   E_{kn}-E_{k'n}=E_0+\frac{2\pi}{t}m,$$
with $m\in\mathbb{Z}$ and $E_0$ being a fixed number. However, as energy levels are time-independent, $m$ must be $0$. Moreover, as $r_{kk}(t)=1$ one has $E_0=0$. The no-decoherence condition thus reads
\begin{equation}	\label{eq:no-decoherence_condition}
	\sum\limits_{i=1}^{K}[ (-1)^{k_i}-(-1)^{k'_i} ]\sum\limits_{j=1}^{N}g_{ij}(-1)^{n_j}=0.
\end{equation}

Let us now first analyze a limiting case of collective decoherence, in which all register spins are coupled to the environment in the same way, that is, the  coupling coefficients $g_{ij}$ are $i$-independent: $g_{ij}=g_j$. Utilizing this fact condition \eqref{eq:no-decoherence_condition} simplifies to
\begin{equation} \label{eq:collective_decoherence}
	\sum\limits_{i=1}^{K}(-1)^{k_i}=\sum\limits_{i=1}^{K}(-1)^{k'_i}.
\end{equation}
Eq. \eqref{eq:collective_decoherence} states that coherence between $|k\rangle$ and $|k'\rangle$ is preserved if binary representation of $k,\, k'$ have equal number of zeros.

\section{General DFSs existence conditions}	\label{sec:general_case}

In the previous section we have obtained the no-decoherence condition in the case of fully symmetrical register-environment interaction under register spin permutations. However, such assumption is not a very realistic one. The question we pose now is how much this symmetry can be perturbed in order to still preserve existence of DFSs. To simplify the discussion let us introduce the following matrices $A$ and $B$  with matrix elements
\begin{equation}\label{}
    A_{ki}=(-1)^{k_i}\ , \ \ \ B_{jn}=(-1)^{n_j}   .
\end{equation}
Moreover, let $S = AGB$ ($G$ is the interaction matrix), that is,
\begin{equation}\label{}
    S_{kn}=\sum\limits_{i=1}^{K}(-1)^{k_i}\cdot\sum\limits_{j=1}^{N}(-1)^{n_j}g_{ij}.
\end{equation}
Using the above definitions condition \eqref{eq:no-decoherence_condition} translates into
\begin{equation}\label{ND}
    S_{kn} - S_{k'n}=0  
\end{equation}
for all $n$.  Now, we determine what constraints must be imposed on $G$ such that the dynamics will not cause decoherence between states $|k\rangle,\, |k'\rangle$. In this paper we consider only the case when $k$ and $k'$ differ by at most two digits in their binary representations. Let $k_{l_1} \neq k'_{l_1}$ or $k_{l_2} \neq k'_{l_2}$. Under such assumptions the no-decoherence condition (\ref{ND}) simplifies to
\begin{equation}	\label{eq:simplified_condition}
  X_{l_1} h_{l_1}^{(n)}  + X_{l_2} h_{l_2}^{(n)}  = 0 ,
\end{equation}
where
\begin{equation}\label{}
    X_{l_1} = (-1)^{k_{l_1}}-(-1)^{k'_{l_1}} \ , \ \ \ X_{l_2} = (-1)^{k_{l_2}}-(-1)^{k'_{l_2}}  ,
\end{equation}
and
\begin{equation}\label{}
    h_i^{(n)}=\sum\limits_{j=1}^{N} (-1)^{n_j}g_{ij} .
\end{equation}
Since $X_{l_\alpha} = 0,\, \pm 2$ the following four cases should be considered:
\begin{enumerate}
	\item \label{it:equal_signs}
		$X_{l_1} = X_{l_2} \neq 0,$ 
	\item	\label{it:different_signs}
		$X_{l_1} = -X_{l_2} \neq 0 ,$ 
	\item
		\begin{enumerate}
			\item	\label{it:zero_nonzero}
           $X_{l_1}=0$ and $X_{l_2} \neq 0 ,$
			\item \label{it:nonzero_zero}
           $X_{l_1} \neq 0$ and $X_{l_2} = 0 . $
		\end{enumerate}
\end{enumerate}
We exclude the trivial case $X_{l_1}=X_{l_2}$ since it implies $k=k'$.

In order to establish the main result we shall use the following simple
\begin{lem*}	\label{lem:random_walk}
Let $x_i \in \mathbb{R}$ and $\sum_i (-1)^{n_i} x_i=0$ for all $n_i \in \{0,1\}$. Then $x_i=0$ for all $i$.
\end{lem*}
\begin{proof}
Taking $n_i = 0$ one has $\sum_i x_i = 0$. Now, let only only one (say for $i=l$) $n_i =1$, that is, $n_i=0$ for $i \ne l$. One has $x_l  - {\sum_i}' x_i=0$, where ${\sum_i}' x_i = \sum_i x_i - x_l$. Hence
$$   x_l  + {\sum_i}' x_i=0 , \ \ \  x_l - {\sum_i}' x_i=0, $$
which gives $x_l=0$ and it ends the proof since $l$ is arbitrary.
\end{proof}

We can now  present solutions to all four cases:
\begin{enumerate}
	\item  $X_{l_1} = X_{l_2} \neq 0$ if and only if $k_{l_1} = k_{l_2} \neq k'_{l_1} = k'_{l_2}$.  Condition \eqref{eq:simplified_condition} now simplifies to
\begin{equation}\label{}
  h_{l_1}^{(n)}+h_{l_2}^{(n)} =  \sum_j (-1)^{n_j} [ g_{l_1 j} + g_{l_2 j}] = 0  ,
\end{equation}
and Lemma implies $g_{l_1 j} =- g_{l_2 j}$, that is $l_2$-nd row of $G$ is equal to $l_1$-th row multipied by $-1$.

	\item  $X_{l_1} = X_{l_2} \neq 0$ if and only if $k_{l_1} = k'_{l_2}$ and $k_{l_2} = k'_{l_1}$. Condition \eqref{eq:simplified_condition} now simplifies to
\begin{equation}\label{}
  h_{l_1}^{(n)}-h_{l_2}^{(n)} =  \sum_j (-1)^{n_j} [ g_{l_1 j} - g_{l_2 j}] = 0  ,
\end{equation}
and Lemma implies $g_{l_1 j} =g_{l_2 j}$, that is $l_1$ and $l_2$ rows $G$ coincide.

\item
		Both cases are possible only if $k$ and $k'$ differ by only one digit in binary representations:
		\begin{enumerate}
			\item one has $h_{l_2}^{(n)} =  \sum_j (-1)^{n_j} g_{l_2 j} = 0$ and hence $g_{l_2 j} = 0$,

			\item one has $h_{l_1}^{(n)} =  \sum_j (-1)^{n_j} g_{l_1 j} = 0$ and hence $g_{l_1 j} = 0$.
				
		\end{enumerate}
	We conclude that both \ref{it:zero_nonzero} and \ref{it:nonzero_zero} require vanishing of one row of the interaction matrix $G$.
\end{enumerate}
Note, that introducing a matrix $D_l$ defined by $[D_l]_{ij} = \delta_{il}\delta_{jl}$ one has
\begin{equation}\label{}
    [D_l G]_{ij} = \delta_{li} g_{lj} ,
\end{equation}
and hence the relations between rows of $G$ may be reformulated in terms of $D_l$.
The above results are summarized in the table below.

\newsavebox\equal
\begin{lrbox}{\equal}
	\begin{minipage}{0.3\textwidth}
    \begin{equation*}
			\begin{cases}
				|k\rangle=|\ldots 1 \ldots 1 \ldots\rangle	\\
				|k'\rangle=|\ldots 0 \ldots 0 \ldots\rangle
			\end{cases}
		\end{equation*}
  \end{minipage}
\end{lrbox}

\newsavebox\different
\begin{lrbox}{\different}
	\begin{minipage}{0.3\textwidth}
    \begin{equation*}
			\begin{cases}
				|k\rangle=|\ldots 0 \ldots 1 \ldots\rangle	\\
				|k'\rangle=|\ldots 1 \ldots 0 \ldots\rangle
			\end{cases}
		\end{equation*}
  \end{minipage}
\end{lrbox}

\newsavebox\pierwszy
\begin{lrbox}{\pierwszy}
	\begin{minipage}{0.3\textwidth}
    \begin{align*}
			\begin{cases}
				|k\rangle=|\ldots 0 \ldots 0 \ldots\rangle	\\
				|k'\rangle=|\ldots 0 \ldots 1 \ldots\rangle
			\end{cases}		\\
			\begin{cases}
				|k\rangle=|\ldots 1 \ldots 1 \ldots\rangle	\\
				|k'\rangle=|\ldots 1 \ldots 0 \ldots\rangle
			\end{cases}
		\end{align*}
  \end{minipage}
\end{lrbox}

\newsavebox\drugi
\begin{lrbox}{\drugi}
	\begin{minipage}{0.3\textwidth}
    \begin{align*}
			\begin{cases}
				|k\rangle=|\ldots 0 \ldots 0 \ldots\rangle	\\
				|k'\rangle=|\ldots 1 \ldots 0 \ldots\rangle
			\end{cases}	\\
			\begin{cases}
				|k\rangle=|\ldots 0 \ldots 1 \ldots\rangle	\\
				|k'\rangle=|\ldots 1 \ldots 1 \ldots\rangle
			\end{cases}
		\end{align*}
  \end{minipage}
\end{lrbox}

\begin{table}	
	\centering
	\begin{tabular}{  c  c  p{3,5cm}  p{4cm} }
		\toprule
			Case & Form of state vectors & No--decoherence condition & Symmetry of $G$\\
		\midrule
			\ref{it:equal_signs} &	\usebox{\equal} & $h_{l_1}^{(n)} = -h_{l_2}^{(n)}$ & $[D_{l_1} + D_{l_2}]G = 0$  \\
			\ref{it:different_signs} &	\usebox{\different} & $h_{l_1}^{(n)} = h_{l_2}^{(n)}$ & $[D_{l_1} - D_{l_2}]G = 0$ \\
			\ref{it:zero_nonzero} &	\usebox{\pierwszy} & $h_{l_2}^{(n)}=0$ & $D_{l_2}G = 0$ \\
			\ref{it:nonzero_zero} &	\usebox{\drugi} & $h_{l_1}^{(n)}=0$ & $D_{l_1}G = 0$ \\
		\bottomrule
	\end{tabular}
\end{table}

\section{Discussion}	\label{sec:discussion}

We have shown that conditions for decoherence-free dynamics within a model defined by (\ref{eq:hamiltonian_form}) can be described by the symmetries of the interaction matrix $G$.

Let us now explore the structure of register Hilbert space if the structure of $G$ allows existence of DFSs. Let us first observe that within this model any decoherence-free subspace has its conjugate one. To express it formally, assume that there exists a DFS $\mathcal{D}=\spanned\{ |k\rangle \}_{k\in\Delta}$, where $\Delta$ is a set of indices. It turn out that a subspace $\mathcal{D'}=\spanned\{ |2^K-k-1\rangle \}_{k\in\Delta}$ provides another DFS. To prove this statement recall that binary form of $2^K-k-1$ can be obtained from that of $k$ by binary complement operation in which all zeros change to ones and vice versa. Combining this fact with definition of the matrix $S$ yields that $S_{2^K-k-1,n}=-S_{kn}$, which immediately implies the result.

In collective decoherence case we have found that coherence is preserved between vectors $|k\rangle,\, |k'\rangle$ if $k$ and $k'$ have the same number of zeros (denote it by $l$) in their binary representations.   It is  easy to see that a subspace $\mathcal{H}_{\mathcal{R},l}\subset\mathcal{H_R}$ spanned by all vectors of the same $l$ is a DFS of dimension $\dim\mathcal{H}_{\mathcal{R},l}=\binom{K}{l}$. Utilizing Stirling's approximation it can be proved that the dimension of the biggest DFS grows as $\frac{2^K}{\sqrt{K}}$ with the size of the register. Moreover, as $\sum\limits_{l=0}^{K}\binom{K}{l}=2^K$ the whole register Hilbert space decomposes into a direct sum of decoherence-free subspaces:
\begin{equation}
	\mathcal{H_R}=\bigoplus\limits_{l=0}^{K}\mathcal{H}_{\mathcal{R},l}.
\end{equation}
The more general setting considered in Section \ref{sec:general_case} is less trivial. In each case one pair of vectors constituting a two dimensional DFS can be chosen in $2^{K-2}$ ways as only two of $K$ digits in binary form of $k,\, k'$ are fixed. This implies that in both cases \ref{it:equal_signs} and \ref{it:different_signs} there are $2^{K-2}$, while in cases \ref{it:zero_nonzero}, \ref{it:nonzero_zero} there are in total $2^{K-1}$ DFSs. This yields that only in the last two cases $\mathcal{H_R}$ decomposes again into a direct sum of decoherence-free subspaces.

\section*{Acknowledgements}
This work was partially supported by the National Science Center project
DEC-2011/03/B/ST2/ 00136.

\end{document}